\documentclass{llncs}
\usepackage[utf8]{inputenc}

\usepackage{a4wide}
\usepackage{xcolor}
\definecolor{darkgreen}{rgb}{0,0.5,0}
\usepackage[colorlinks=true,citecolor=darkgreen]{hyperref}
\usepackage{amsmath,amssymb,amsfonts}
\usepackage{booktabs}
\usepackage{graphicx}
\usepackage{algorithm} 
\usepackage{algpseudocode} 

\newcommand{\Z}{\ensuremath{\mathbb{Z}}}
\newcommand{\Q}{\ensuremath{\mathbb{Q}}}
 
\usepackage{xcolor}

\newif\ifhidetodos
\hidetodosfalse
\ifhidetodos
\usepackage[final,multiuser]{fixme}
\else
\usepackage[draft,multiuser]{fixme}
\fi
\fxsetup{theme=color,mode=multiuser}
\FXRegisterAuthor{jb}{ejb}{\color{red}Jan}
\FXRegisterAuthor{jh}{ejh}{\color{orange}Junaid}
\FXRegisterAuthor{jl}{ejl}{\color{blue}Jason}
\FXRegisterAuthor{mn}{emn}{\color{teal}Michael}
\FXRegisterAuthor{av}{eav}{\color{green}Tony}

\title{Finding twin smooth integers by solving Pell equations}
\author{
Jan Buzek \inst{1}\thanks{This work was started as an undergraduate research project at the Washington eXperimental Mathematics Lab (WXML) of the Department of Mathematics at the University of Washington, Seattle.}\and
Junaid Hasan \inst{1}\and 
Jason Liu \inst{1}\and 
Michael Naehrig \inst{2}\and 
Anthony Vigil \inst{1}
}
\institute{
University of Washington, Seattle, WA, USA \\
\and
Microsoft Research, Redmond, WA, USA
}

\begin{document}

\maketitle

\begin{abstract}
Any pair of consecutive $B$-smooth integers for a given smoothness bound $B$ corresponds to a solution $(x,y)$ of the equation $x^2 - 2\Delta y^2 = 1$ for a certain square-free, $B$-smooth integer $\Delta$ and a $B$-smooth integer $y$. This paper describes algorithms to find such twin $B$-smooth integers that lie in a given interval by using the structure of solutions of the above Pell equation.

The problem of finding such twin smooth integers is motivated by the quest for suitable parameters to efficiently instantiate recent isogeny-based cryptosystems. While the Pell equation structure of twin $B$-smooth integers has previously been used to describe and compute the full set of such pairs for very small values of $B$, increasing $B$ to allow for cryptographically sized solutions makes this approach utterly infeasible.

We start by revisiting the Pell solution structure of the set of twin smooth integers. Instead of using it to enumerate all twin smooth pairs, we focus on identifying only those that lie in a given interval. This restriction allows us to describe algorithms that navigate the vast set of Pell solutions in a more targeted way. Experiments run with these algorithms have provided examples of twin $B$-smooth pairs that are larger and have smaller smoothness bound $B$ than previously reported pairs. Unfortunately, those examples do not yet provide better parameters for cryptography, but we hope that our methods can be generalized or used as subroutines in future work to achieve that goal. 
\end{abstract}

\section{Introduction}\label{sec:intro}
Twin smooth integers are pairs of consecutive smooth integers and have been studied for various reasons at least since the 18th century. They have been mentioned by Gauss in the context of computing logarithms of integers, and they can be used to find particular solutions to certain diophantine equations~\cite{lehmer}.  

\begin{definition}[Smooth and twin smooth integers]
Let $B>1$ and $m\in \Z$. The integer $m$ is called \emph{$B$-smooth} if no prime factor of $m$  is larger than $B$. The pair $(m, m+1)$ is called a pair of \emph{twin $B$-smooth integers} if $m\cdot(m+1)$ is $B$-smooth.
\end{definition}
St{\o}rmer~\cite{stormer} proved that for a fixed smoothness bound $B>1$, there is a finite number of twin $B$-smooth integers and that all such pairs can be described via solutions of Pell equations of the form $x^2 - D y^2 = 1$. 
Lehmer~\cite{lehmer} improved upon St{\o}rmer's work in utilizing the structure of the set of solutions to the Pell equation and strongly reduced the number of Pell equations that must be solved in order to find all twin $B$-smooth pairs. Lehmer computed the full set for $B=41$. Unfortunately, the number of Pell equations that must be solved is still exponential in $B$, and therefore finding all $B$-smooth pairs is computationally prohibitive for $B$ larger than 200.

In this paper, we focus on the problem of finding some B-smooth pairs that lie in a certain interval of possibly very large integers, rather than finding all pairs; for example, given a smoothness bound $B$ as well as a bitlength $b$, finding twin $B$-smooth integers that lie in the interval $[2^{b-1}, 2^b)$, i.e. finding such pairs $(m, m+1)$, for which $m$ can be written with $b$ bits. A slight variant of the problem only assumes the bitlength $b$ is fixed and then asks to find twin $B$-smooth integers with $b$ bits and the lowest possible smoothness bound $B$.

\subsubsection{Cryptographic applications.}
The latter variant of the problem is motivated by applications in cryptography. It has occurred recently in the context of selecting parameters for instantiating public key cryptography that uses isogenies between elliptic curves. For example, the digital signature scheme SQISign~\cite{sqisign} requires a large prime number $p$ such that $p-1$ and $p+1$ have large smooth factors. For security reasons, the prime $p$ (and therefore the integer $m$) needs to have a specific minimal size, usually given as a minimal number of bits, to ensure that the schemes cannot be easily broken. 

If one aims to find $p$ such that both $p-1$ and $p+1$ are fully smooth, the problem is to find twin smooth integers $m = (p-1)/2$ and $m+1 = (p+1)/2$ with the additional condition that they sum to a prime number $p = m + (m+1)$. While this was an interesting scenario for the (now broken) key exchange scheme B-SIDH~\cite{b-sidh}, it is not strictly necessary. Allowing a small non-smooth cofactor in $p-1$ and $p+1$ often leads to a lower smoothness bound and a more efficient implementation of the cryptosystem. Still, algorithms for finding (fully) $B$-smooth twins, typically with much shorter bitsize, can be used to construct cryptographically sized primes.

This motivates to look for twin smooths $(m, m+1)$ in a prescribed interval, or more specifically with a given bitlength $b$, i.e. $\lfloor\log_2(m)\rfloor + 1 =b$ Furthermore, the smoothness bound $B$ for which a suitable twin $B$-smooth pair $(m, m+1)$ of the given size needs to be found, has a direct impact on the efficiency of cryptographic primitives. The smaller $B$ is, the more efficient the protocol becomes.

\subsubsection{Algorithms for finding twin smooth integers.}
As mentioned above, Lehmer's paper~\cite{lehmer} suggests to enumerate all twin $B$-smooth pairs for fixed $B$ by solving the Pell equation $x^2 - 2\Delta y^2 = 1$ for an exponential number of values for $\Delta$. This becomes infeasible for even moderately large smoothness bound $B$. Luca and Najman~\cite{luca-najman,luca-najman-err} succeeded in computing all twin 100-smooth pairs with a modified version of Lehmer's approach. This computation was very costly because of the large number of Pell equations to be solved and the solutions' potentially huge sizes. Costello \cite[\S 5.3]{b-sidh} used this approach to search for B-SIDH parameters and reports to have solved all $2^{30}$ Pell equations for $B=113$ (the 30-th prime) in a large computational effort, discovering the largest one to be a 76-bit $113$-smooth pair.
An interesting algorithm that approximates this problem and appears to be much more efficient is presented in~\cite{Conrey}.

\subsubsection{The Conrey-Holmstrom-McLaughlin approach.} 
Conrey, Holmstrom, and McLaughlin propose an algorithm in \cite{Conrey} that quickly generates a large set of moderately-sized twin smooth integers for a given smoothness bound $B$. Experiments for relatively small $B$ demonstrate that their method succeeds in generating almost all $B$-smooth pairs, only missing a few that are often picked up when running the algorithm for a slightly larger $B$. The above example of finding 100-smooth twins was easily run, but missed 37 pairs of the Luca-Najman results. A significantly more costly search for 200-smooth twins recovered all but one of those. 

The algorithm works by observing that two twin smooth integer pairs, $(m, m+1)$ and $(M, M+1)$ with $m < M$, generate a new pair of twin smooths $(\mu, \mu+1)$ if $\mu$ is an integer and $\frac{\mu}{\mu+1} = \frac{m}{m+1}\frac{M+1}{M}$. This can be used to iteratively expand any set of twin smooths. Begin with the set $S$ of integers $m < B$, though $S$ can be any set of integers $m$ such that $(m, m+1)$ forms a smooth pair. At each iteration, take all pairs $m < M$ in $S$ and collect all new $\mu$ of the above form in a set $S'$; set $S = S \cup S'$; if $S' = \phi$, the algorithm stops. The largest member of $S'$ is bounded by the square of the largest member of $S$, so the best case scenario is doubling the number of bits at each iteration. 

While much more efficient in constructing most twin smooth pairs, this algorithm seems to quickly become infeasible when aiming for large integers that require a larger smoothness bound. The main reason is that the cardinality of $S$ balloons rapidly, which makes successive iterations infeasible before reaching a large enough bitlength. Nonetheless, this algorithm is a powerful tool for quickly generating large sets for experimentation, and an interesting avenue for future research, as demonstrated recently in \cite{CHM4SQISign}, where it was used as a subroutine for successfully constructing parameters for SQISign. The paper also reports that the authors ran the full algorithm for $B=547$. The largest $547$-smooth twin smooth pair found in these experiments has 122 bits. 

Moreover, the following algorithms have been used in the context of finding cryptographic parameters. They do not attempt to compute the complete finite set of twin $B$-smooth integers, but instead have the goal of generating just enough ``good'' examples of a specific bit size. Ultimately the goal is to find a pair with prime sum, but the probability of this happening is often large enough within a large enough set of pairs. Specifically, the Prime Number Theorem gives an estimate that random numbers in the cryptographic range (around 256 bits) that we are looking for have a chance on the order of one in a hundred of being prime. In comparison, the chances of random consecutive integers in that range being B-smooth for B at most $2^{19}$ are tiny. Thus assuming there is no interaction between smoothness of $(m, m+1)$ and primality of $2m+1$, an algorithm that efficiently finds smooth pairs and then checks if their sum is prime is an efficient way to find cryptographic parameters, whereas an algorithm that checks whether $(\frac{p-1}{2}, \frac{p+1}{2})$ is a smooth pair for large primes $p$ is not. Therefore, the main task is to find the twin smooth pairs.

\subsubsection{Sieving.}
A simple algorithm to search a given interval is to use a variant of the sieve of Eratosthenes to list all $B$-smooth integers $m$ in the interval and check for each whether $m+1$ is also smooth. This is inefficient because for large $m$, the probability that $m+1$ is also smooth is very low and the chance to encounter such pairs is so low that searchable intervals are too small and likely do not even contain a solution or the search space is too large to be covered. Nevertheless, plain sieving has been successfully used as a subroutine in other methods. For example, some search algorithms increase the probability of hitting twin smooth integers by using polynomial parameterizations.

\subsubsection{Sieving with $x^n - 1$.}\label{sec:Poly-sieve}
There are various polynomial parameterizations leading to twin smooth integers that can be exploited by sieving for patterns of smaller smooth integers. A simple polynomial that was found to work well in~\cite{pte}, \cite{b-sidh}, and \cite{CHM4SQISign} is $f(x) = x^n - 1$ for small values of $n$. For example, if $n = 2$, then $(m, m+1) = (f(x), f(x) + 1)$ is a $B$-smooth pair when $(x-1)$, $x$, and $(x+1)$ are all smooth by the factorization $x^2 - 1 = (x - 1)(x + 1)$. One can then use the sieve to find three consecutive smooth integers $x-1, x, x+1$ and obtain the twin pair $(x^2-1, x^2)$ of twice as many bits. Because each of these factors are about the square root of the resulting smooth pair, it is easier to find the smooth triple of small integers. This phenomenon extends similarly to larger values of $n$ such as $4$ or $6$. We give further explanation of the reasons that certain polynomials work well in \S\ref{sec:pell} (see in particular the polynomials $p_n(x)$).

\subsubsection{Sieving with Solutions to the Prouhet-Tarry-Escott Problem.}
The Prouhet-Tarry-Escott (PTE) method is a generalization of the above idea that finding smooth pairs parameterized by polynomials can be easier because the polynomial factors of $(m, m+1)$ required to be smooth are much smaller than the $m$ and $m+1$ themselves. Solutions to the PTE problem correspond to polynomials $f(x)$ such that both $f(x)$ and $f(x) + 1$ factor into linear terms over $\mathbb{Z}$. Given such polynomials, we can search for pairs $(m = f(x), m + 1 = f(x) + 1)$ by checking for smoothness of all of the linear factors of $f$ and $f + 1$ for the same value of $x$. This is again done using a sieve. Because smaller integers are much more likely to be smooth, it is often more efficient to search for many small integers that are all smooth than for a single large pair directly. The PTE method has found many of the best-known pairs of cryptographic size; for details see \cite{pte}. However, this method is restricted in that only a finite number of solutions to the PTE problem are known, and all have degree at most $12$. Furthermore, it only finds pairs of this special form, and it is unclear what type of pairs might be missed.

\subsubsection{Extended Euclidean Algorithm.}
 Another improvement, discussed for example in ~\cite{b-sidh}, uses the extended Euclidean algorithm. It chooses two coprime, $B$-smooth integers $\alpha$, $\beta$ roughly the square root of the size of the desired pair and runs the Euclidean algorithm. This produces an equality of the form $\alpha  s + \beta t = 1$, and we take $| \alpha s |,  | \beta t |$ as our candidate smooth pair. The benefit is that the portions $\alpha$ and  $\beta$ of these numbers are already known to be smooth, increasing the probability that $|\alpha \beta s t|$ will be smooth overall. For more details and a detailed analysis of the probabilities for various sizes of $B$ and $m$, see~\cite{pte} and~\cite{b-sidh}.

\subsubsection{Contributions.}
In this paper, we describe how Pell equations can be used to search for large twin-smooth pairs. Pell equations have been believed to be mainly a theoretical tool useful for ensuring that all pairs have been found; we show that they can also provide an efficient way of searching. We give an algorithm generalizing and explaining the polynomial sieving described above and another based on heuristically choosing which Pell equations to solve. We give experimental evidence for our algorithms through the pairs that they found in our search of the interval $[2^{240}, 2^{256}]$, including pairs that are the largest known for their smoothness bound.  

\subsubsection{Outline.}
In Section \ref{sec:pell}, we explore the structure of the solutions to the Pell equation, deducing some results on explicit polynomial parameterizations from the work of Lehmer. Then we specialize to the case of searching for pairs in a specified large interval and describe how solving Pell equations can be improved tremendously with this assumption. In Section \ref{sec:twins}, we apply this assumption and our results on the structure of solutions to choosing which Pell equations and which of their solutions to consider.

\section{Solutions of the Pell equation}\label{sec:pell}
For a fixed, non-square positive integer $D$, the Pell equation is a Diophantine equation of the form
\begin{equation}\label{eqn:pell}
    x^2 - Dy^2 = 1
\end{equation}
and its solutions are pairs of integers $(x, y)$ that satisfy~\eqref{eqn:pell}. The Pell equation was studied in the seventh century by Brahmagupta, and since at least the seventeenth century by mathematicians including Euler, Fermat, and Lagrange. We give a brief overview of the results that interest us, generally omitting proofs as these are readily available in resources such as \cite{lehmer}, \cite{jacobson-williams} and \cite{lenstra02}.

Equation~\eqref{eqn:pell} always has the trivial solutions $(\pm 1, 0)$, but our main interest lies in describing all non-trivial solutions $(x, y)$ and without loss of generality, we focus on those with $x, y > 0$. The following well-known result describes the structure of all such solutions (see, e.g.~\cite[Chapter 1]{jacobson-williams}, \cite[\S1]{lenstra02} and \cite[\S2]{lehmer}).

\begin{theorem}\label{thm:Pell_Solutions}
Let $D>0$ be an integer that is not a perfect square. Then, the Pell equation $x^2 - Dy^2 = 1$ has an infinite number of integer solutions $(x,y)$ with $x,y > 0$. The solution with the smallest value of $x+y\sqrt{D}$ among them is called the \emph{fundamental solution} and is denoted by $(x_1, y_1)$. All nontrivial solutions as above can be written as $(x_n, y_n)$ for some integer $n \geq 1$, where
\begin{equation}\label{eqn:solutions}
    x_n + y_n \sqrt{D} = (x_1 + y_1 \sqrt{D})^n.
\end{equation}
The solution $(x_n, y_n)$ is called the \emph{$n$-th solution}.
\end{theorem}

\subsection{Expressing solutions in terms of the fundamental solution}
Lehmer~\cite[\S2]{lehmer} describes all solutions of \eqref{eqn:pell} in terms of the Lucas function $U_n$. Let $\alpha = x_1 + y_1 \sqrt{D}$ and $\beta = x_1 - y_1\sqrt{D}$. Then, $\alpha\beta=1$ and for all $n \ge 1$, we get $\alpha^n+\beta^n = 2x_n$, $\alpha^n-\beta^n = 2y_n\sqrt{D}$. The Lucas function $U_n$ is defined as 
\begin{equation}\label{eqn:Un}
    U_n = (\alpha^n-\beta^n)/(\alpha-\beta) = y_n/y_1.
\end{equation}
Hence $U_1 = 1$ and the above identities yield recursive formulas to compute $U_n$ for all $n$, since $x_{2n} = 2x_n^2 -1$, $U_{2n} = 2x_n U_n$, $x_{n+1} = x_n x_1 + Dy_n y_1 = x_n x_1 + U_n(x_1^2 - 1)$ and $U_{n+1} = x_n + x_1 U_n$. Therefore, $U_n$ assumes integer values and $y_1 \mid y_n$ for all $n \ge 1$. Also, writing out \eqref{eqn:solutions} gives the explicit formulas
\begin{equation}\label{eqn:xnUn}
    x_n = \sum_{i=0}^{\lfloor\frac{n}{2}\rfloor} \binom{n}{2i} (Dy_1^2)^ix_1^{n-2i},\quad 
    U_n = \sum_{i=0}^{\lceil\frac{n}{2}\rceil-1} \binom{n}{2i+1} (Dy_1^2)^ix_1^{n-2i-1}.
\end{equation}
Note that we can replace $Dy_1^2 = x_1^2 - 1$ via the Pell equation to obtain univariate polynomials in $x_1$. We write these as polynomials $p_n(x)\in \Z[x]$ such that $x_n = p_n(x_1)$, i.e. we have $p_n= \sum_{i=0}^{\lfloor\frac{n}{2}\rfloor} \binom{n}{2i} (x^2-1)^ix^{n-2i}$ and $\deg(p_n) = n$. Expanding $(x^2-1)^i$ and reordering the terms gives
\begin{equation}\label{eqn:pn}
    p_n(x) = \sum_{i=0}^{\lfloor\frac{n}{2}\rfloor} \sum_{j=0}^i \binom{n}{2i} \binom{i}{j} (-1)^j x^{n-2j}
    = \sum_{j=0}^{\lfloor\frac{n}{2}\rfloor} \left((-1)^j \sum_{i=j}^{\lfloor\frac{n}{2}\rfloor} \binom{n}{2i} \binom{i}{j}\right) x^{n-2j}.
\end{equation}
The first few $p_n$, $n \in \{1,2,\dots,12\}$ are given as follows:
\begin{align*}
    p_1(x) & = x,\\
    p_2(x) & = 2x^2 - 1,\\
    p_3(x) & = x (4x^2 - 3),\\
    p_4(x) & = 8x^4 - 8x^2 + 1,\\
    p_5(x) & = x (16x^4 - 20x^2 + 5),\\
    p_6(x) & = (2x^2 - 1) (16x^4 - 16x^2 + 1),\\
    p_7(x) & = x (64x^6 - 112x^4 + 56x^2 - 7),\\
    p_8(x) & = 128x^8 - 256x^6 + 160x^4 - 32x^2 + 1,\\
    p_9(x) & = x(4x^2 - 3)(64x^6 - 96x^4 + 36x^2 - 3),\\
    p_{10}(x) & = (2x^2 - 1) (256x^8 - 512x^6 + 304x^4 - 48x^2 + 1),\\
    p_{11}(x) & = x (1024x^{10} - 2816x^8 + 2816x^6 - 1232x^4 + 220x^2 - 11),\\
    p_{12}(x) & = (8x^4 - 8x^2 + 1) (256x^8 - 512x^6 + 320x^4 - 64x^2 + 1).
\end{align*}
Similarly, Equation~\eqref{eqn:xnUn} gives an explicit formula for polynomials $u_n(x)\in \Z[x]$ of degree $n-1$ such that $U_n = u_n(x_1)$, namely $u_n(x) = \sum_{i=0}^{\lceil\frac{n}{2}\rceil-1} \binom{n}{2i+1} (x^2-1)^ix^{n-2i-1}$. As above, by expanding $(x^2-1)^i$ and reordering, $u_n(x)$ can be rewritten as
\begin{equation}\label{eqn:un}
    u_n(x) = \sum_{j=0}^{\lceil\frac{n}{2}\rceil - 1} \left((-1)^j \sum_{i=j}^{\lceil\frac{n}{2}\rceil - 1} \binom{n}{2i+1} \binom{i}{j}\right) x^{n-2j-1}.
\end{equation}
Here are the first $12$ polynomials $u_n(x)$:
\begin{align*}
    u_1(x) & = 1,\\
    u_2(x) & = 2x,\\
    u_3(x) & = (2x - 1)(2x + 1),\\
    u_4(x) & = 4x(2x^2 - 1),\\
    u_5(x) & = (4x^2 - 2x - 1)(4x^2 + 2x - 1),\\
    u_6(x) & = 2x(2x - 1)(2x + 1)(4x^2 - 3),\\
    u_7(x) & = (8x^3 - 4x^2 - 4x + 1)(8x^3 + 4x^2 - 4x - 1),\\
    u_8(x) & = 8x(2x^2 - 1)(8x^4 - 8x^2 + 1),\\
    u_9(x) & = (2x - 1)(2x + 1)(8x^3 - 6x - 1)(8x^3 - 6x + 1),\\
    u_{10}(x) & = 2x(4x^2 - 2x - 1)(4x^2 + 2x - 1)(16x^4 - 20x^2 + 5),\\
    u_{11}(x) & = (32x^5 - 16x^4 - 32x^3 + 12x^2 + 6x - 1)(32x^5 + 16x^4 - 32x^3 - 12x^2 + 6x + 1),\\
    u_{12}(x) & = 4x(2x - 1)(2x + 1)(2x^2 - 1)(4x^2 - 3)(16x^4 - 16x^2 + 1).
\end{align*}
Using the above polynomials, one can easily write down the $n$-th solution to the Pell equation~\eqref{eqn:pell} in terms of the fundamental solution $(x_1, y_1)$ as 
\begin{equation}
    (x_n, y_n) = (p_n(x_1), y_1u_n(x_1)). 
\end{equation}

Because the $p_n$ and $u_n$ are directly derived from \eqref{eqn:solutions} and \eqref{eqn:pell}, we obtain the following lemma stating some divisibility properties.

\begin{lemma}\label{lem:pnun}
Let $p_n(x), u_n(x)\in \Z[x]$ be the above polynomials such that $x_n = p_n(x_1)$, $U_n = u_n(x_1)$ for $0 < n \in \Z$, $\deg(p_n) = n$ and $\deg(u_n) = n-1$.   
\begin{enumerate}
    \item[(a)] If $n$ is odd, then $x \mid p_n(x)$. 
    \item[(b)] If $n = km$ for $0 < k, m \in \Z$, then $p_n(x) = p_m(p_k(x))$ and therefore, if $m$ is odd, $p_k(x) \mid p_n(x)$. 
    \item[(c)] Furthermore, $u_n(x) = u_k(x) u_m(p_k(x))$. In particular, if $k\mid n$, then $u_k \mid u_n$.
    \item[(d)] The leading coefficients of $p_n(x)$ and $u_n(x)$ are both equal to $2^{n-1}$.
    \item[(e)] If $n=2^k$, $k\geq 1$, then $p_n(x)$ is irreducible in $\Z[x]$.
\end{enumerate} 
\end{lemma}
\begin{proof}
Statement (a) follows from the explicit formula for $x_n$ in \eqref{eqn:xnUn}. Note that if $n$ is odd, the exponent $n-2i$ of $x_1$ is always at least $1$. 

If $n$ factors as $n=km$, then the exponentiation in Equation~\eqref{eqn:solutions} can be split up into two parts, $x_n + y_n \sqrt{D} = (x_1 + y_1 \sqrt{D})^n = ((x_1 + y_1 \sqrt{D})^k)^m = (x_k + y_k \sqrt{D})^m$. Writing out the last exponentiation by $m$ in terms of $x_k$ and $y_k$ explicitly gives $x_n + y_n\sqrt{D} = p_m(x_k) + u_m(x_k) y_k \sqrt{D} = p_m(x_k) + u_m(x_k) y_1 u_k(x_1) \sqrt{D}$. Now writing $x_k$ in terms of $x_1$ gives $p_m(p_k(x_1)) + u_m(p_k(x_1)) y_1 u_k(x_1) \sqrt{D}$. The exact same calculations can be carried out as polynomials in $x$ instead of expressions in $x_1$, which gives $x_n(x) + y_n(x)\sqrt{D} = p_m(p_k(x)) + y_1 u_k(x) u_m(p_k(x))\sqrt{D}$, which proves (b) and (c).

Expression~\eqref{eqn:pn} shows that the leading coefficient of $p_n(x)$ is $\sum_{i=0}^{\lfloor\frac{n}{2}\rfloor}\binom{n}{2i} = \binom{n}{0} + \binom{n}{2} + \dots + \binom{n}{2\lfloor\frac{n}{2}\rfloor} = 2^{n-1}$. From~\eqref{eqn:un}, one obtains the leading coefficient of $u_n(x)$ as $\sum_{i=0}^{\lceil\frac{n}{2}\rceil - 1}\binom{n}{2i+1} = \binom{n}{1} + \binom{n}{2} + \dots + \binom{n}{2\lceil\frac{n}{2}\rceil - 1} = 2^{n-1}$, and thus (d) holds. 

To prove (e), we first take a look at the coefficients of $p_{2^k}(x)$ when we write it as $p_{2^k}(x) = \sum_{j=0}^{2^{k-1}} a_j x^{2^k - 2j} = a_0x^{2^k} + a_1x^{2^k-2} + \dots + a_{2^{k-1}-1}x^2 + a_{2^{k-1}}$ with 
$$a_j = (-1)^j \sum_{i=j}^{2^{k-1}} \binom{2^k}{2i} \binom{i}{j}$$ 
as in \eqref{eqn:pn}. Note that $\lfloor n/2\rfloor = 2^{k-1}$. The constant coefficient is $a_{2^{k-1}} = 1$ for $k>1$ and $a_{2^{k-1}} = a_1 = -1$ for $k=1$. As we have seen in (d), the leading coefficient is $a_0 = \sum_{i=0}^{2^{k-1}}\binom{2^k}{2i} = 2^{2^k-1}$.

We now show by induction on $k$ that, for all coefficients $a_j$ with $1\leq j \leq 2^{k-1}$, $2^{2^k-2j}$ divides $a_j$. For the polynomial $p_2(x) = 2x^2 -1$ ($k=1$), this is true. For $k = 2$, we have $p_4(x) = 8x^4 - 8x^2 + 1$ and $a_1 = -8$ is divisible by $4 = 2^{2^2-2}$. For the induction step, note that $p_{2^{k+1}}(x) = 2(p_{2^k}(x))^2 - 1$ by the recursive formula for $x_{2n}$ discussed right after Equation~\eqref{eqn:Un}. If we write $(p_{2^{k}}(x))^2 = \sum_{j=0}^{2^{k}} b_j x^{2^{k+1} - 2j} = b_0x^{2^{k+1}} + b_1x^{2^{k+1}-2} + \dots + b_{2^{k}-1}x^2 + b_{2^{k}}$, then 
$$
b_j = \sum_{\substack{0\leq l,m \leq 2^{k-1}\\ l+m = j}}a_la_m.
$$
If $lm \neq 0$, then $2^{2^k-2m} \mid a_m$ and $2^{2^k-2l}\mid a_l$, so $2^{2^{k+1}- 2(l+m)}\mid a_la_m$, i.e. if $l+m=j$, then $2^{2^{k+1}- 2j}\mid a_la_m$. If one of $l$ or $m$ is $0$, the term involves $a_0=2^{2^k-1}$. In all $b_j$ with $j>0$, $a_0$ only occurs in products with $a_j$, which appear twice. Therefore, in these sums, the two terms can be combined to $2a_0a_j$ and it follows again that $2^{2^{k+1}-2j} \mid 2a_0a_j$. Overall, this shows that $2^{2^{k+1}-2j} \mid b_j$. Since $p_{2^{k+1}} = 2b_0x^{2^{k+1}} + 2b_1x^{2^{k+1}-2} + \dots + 2b_{2^{k}-1}x^2 + (2b_{2^{k}}-1)$, the divisibility conditions still hold, which concludes our induction proof.  

Now consider the polynomial $q(x) = 2p_{2^k}(x/2)$. The leading coefficient cancels exactly, making $q(x)$ monic. The constant term in $p_{2^k}(x)$ is equal to $1$, which means that the constant term in $q(x)$ is equal to $2$. The property we just proved ensures that all other coefficients of $p_{2^k}(x/2)$ are integers. Thus $q(x) = x^{2^k} + c_1x^{2^k-2} + \dots + c_{2^{k-1}-1}x^2 + 2$ for even integers $c_1, \dots, c_{2^{k-1}-1}$. The Eisenstein irreducibility criterion with the prime $p=2$ shows that $q(x)$ is irreducible in $\Z[x]$. Then, $p_{2^k}(x/2) = q(x)/2$ is irreducible in $\Q[x]$, and so is $p_{2^k}(x)$. Because $p_{2^k}(x)\in \Z[x]$ is primitive, it is also irreducible in $\Z[x]$. 
\qed
\end{proof}
This lemma shows that $p_n(x)$ is irreducible if and only if $n = 2^k$ for some $k\geq 1$.

\subsection{Finding the fundamental solution via continued fractions}\label{sec:cf}
In light of the above, solving the Pell equation is reduced to finding its fundamental solution $(x_1, y_1)$. This can be done via continued fraction expansions as we explain next.

We focus on simple continued fractions (CF) for square roots, which are expressions of the form
$$\sqrt{D} = a_0 + \cfrac{1}{a_1 + \cfrac{1}{a_2 + \cfrac{1}{a_3 + \cfrac{1}{a_4 + \dots}}}}$$
for $D$ a squarefree integer and a sequence of integers $a_0, a_1, \dots$, called the quotients, where $a_i > 0$ for all $i$. 

Every rational number can be expressed in two different ways via CF, and every irrational number has a unique CF representation that is an infinite sequence (\cite{unger}, 2.2); $\sqrt{D}$ falls in this second case.
We form the sequence of convergents $\frac{A_i}{B_i}$ to $\sqrt{D}$ by truncating the CF expansion after a finite number of terms. Specifically, the $i$th convergent is given by
$$a_0 + \cfrac{1}{a_1 + \cfrac{1}{\ddots + \cfrac{1}{a_i}}}.$$
These convergents can also be computed by the recurrence relations
\begin{eqnarray*}
A_i & = a_i A_{i-1} + A_{i-2},\\
B_i & = a_i B_{i-1} + B_{i-2},
\end{eqnarray*}
beginning with the values $A_0 = a_0$, $A_1 = a_1 a_0 + 1$, $B_0 = 1$, and $B_1 = a_1$ (\cite{unger}, 2.4). An important observation is that the numerators and denominators $A_i$ and $B_i$ of these convergents are monotonically increasing; further because the $a_i$ are positive integers, these sequences are bounded below by the Fibonacci sequence.

The quotients $a_i$ can be found almost directly from the definition by recursively reciprocating and rounding down, however, the following recurrence relations (\cite[Theorem 10.19]{Rosen}) give a faster implementation.
\begin{align*}
\alpha_k & = \left\lfloor\frac{P_k+\sqrt{D}}{Q_k}\right\rfloor,\\
P_{k+1} & = a_k Q_{k}-P_{k},\\
Q_{k+1} & = \frac{d - P_{k+1}^2}{Q_k},
\end{align*}
where $Q_0 = 1, P_0 = 0$.

The connection between continued fractions and the Pell equation is given by the following theorem (\cite[Theorem 3.3]{jacobson-williams}).

\begin{theorem}\label{thm:Continued_Fractions}
Let $D>0$ be an integer that is not a perfect square. Then if $(x, y)$ is a solution to the Pell equation $x^2 - Dy^2 = 1$ then $\frac{x}{y}$ is a convergent in the continued fraction expansion for $\sqrt{D}$.
\end{theorem}

This theorem suggests that to find the fundamental solution, we generate the quotients $a_i$ with our recurrence relations, using these to generate $A_i$ and $B_i$ for each $i$, and check whether $A_i$ and $B_i$ satisfy the Pell equation as $x$ and $y$. Because a fundamental solution exists, this algorithm always terminates. We will only be interested in solutions to the Pell equation that lie in a given range; this additional restriction and our bounds on the monotonically increasing sequence $A_i$ will ensure that this algorithm is very fast. Specifically, if we give an upper bound $\Tilde{M}$ on the size of $x$ in the solution to the Pell equation that we want to find, then we can cut off our algorithm after the first convergent $A_i$ is generated that is larger than $\Tilde{M}$. Because the $A_i$ are bounded below by the Fibonnaci sequence, this means that our algorithm will find any solution with $x \leq \Tilde{M}$ in time proportional to log($\Tilde{M}$). In our search for cryptographically large twin smooth pairs by solving Pell equations, $\Tilde{M}$ is $2^{258}$, and a single core running our C-code implementing this algorithm can find any solutions with $x \leq \Tilde{M}$ for tens of thousands of Pell equations a second.

\section{Twin smooth integers as solutions of Pell equations}\label{sec:twins}
Lehmer states the correspondence between $B$-smooth twins and certain solutions to Pell equations in Theorem 1 of \cite{lehmer}. The theorem is more general because it allows a specific finite set of primes as the prime factors of the twin smooth integers in question. But we are mainly interested in capturing $B$-smoothness, so we state it here specialized to consider all primes not exceeding $B$.
\begin{theorem}[Lehmer~\cite{lehmer}, Thm.~1]\label{thm:lehmer}
Let $B>2$ and $0< m\in \Z$. Then $(m, m+1)$ is a pair of twin $B$-smooth integers if and only if there exists a $B$-smooth integer $y$ and a square-free $B$-smooth integer $\Delta\neq 2$ such that $(x_n, y_n)$ with $x_n = 2m+1$, $y_n = y$ is the $n$-th solution to the Pell equation
\begin{equation}\label{eqn:twin_pell}
    x^2 - 2\Delta y^2 = 1
\end{equation}
with $1 \leq n \leq \max\{3, (q+1)/2\}$, where $q$ is the largest prime not exceeding $B$. 
\end{theorem}

\begin{remark}\label{rem:xoddyeven}
If $\Delta$ in Equation~\eqref{eqn:twin_pell} is odd then the equation is a Pell equation in the sense of \S\ref{sec:pell} with a square-free, positive integer $D=2\Delta$. But Theorem~\ref{thm:lehmer} allows $\Delta$ to be even, i.e. in which case $\Delta = 2D$ for an odd, square-free, positive $D$. In this case, a solution $(x_n',y_n')$ to Equation~\eqref{eqn:twin_pell} corresponds to the solution $(x_n=x_n', y_n=2y_n')$ of the Pell equation $x^2 - Dy^2 = 1$, i.e. here $D = \Delta/2$.

Furthermore, Equation~\eqref{eqn:twin_pell} clearly implies that $x_n$ must be odd for any solution $(x_n,y_n)$. So, if $x_n = 2m+1$, then $\Delta y_n^2 = (x_n^2-1)/2 = 2m(m+1)$ is divisible by $4$. Since $\Delta$ is square-free, $y_n$ must be even.
\end{remark}

If $2 = q_1, q_2, \dots, q_t \in \Z$ are all the prime numbers that are not larger than $B$, listed consecutively in ascending order, then Lehmer's theorem shows that in order to find all twin $B$-smooth pairs, it suffices to solve $2^t-1$ Pell equations of the form~\eqref{eqn:twin_pell}, because that is the number of possible values $\Delta$ can take in Theorem~\ref{thm:lehmer}. Indeed, following Lehmer's notation, let $Q$ denote the set of positive $B$-smooth integers, i.e.
\[Q = \{q_1^{\alpha_1}\cdot q_2^{\alpha_2}\cdot \ldots \cdot q_t^{\alpha_t} \mid 0 \le \alpha_i \in \Z,\ 1 \le i \le t\},\]
and let $Q'$ be the subset of square-free elements of $Q$ except $2$, i.e.
\[Q' = \{q_1^{\epsilon_1}\cdot q_2^{\epsilon_2}\cdot \ldots \cdot q_t^{\epsilon_t} \mid \epsilon_i \in \{0,1\},\ 1 \le i \le t\} \setminus \{2\}.\]
Then $Q'$ is the set of possible values of $\Delta$ that can occur in Theorem~\ref{thm:lehmer} and has $2^t-1$ elements.

Theorem~\ref{thm:lehmer} gives an explicit correspondence between twin smooth pairs $(m, m+1)$ and Pell equation coefficients $\Delta\in Q'$ together with solutions $(x,y) \in \Z_{>0}\times Q$. 
Define the set
\[P = \{(\Delta, x, y) \in Q'\times \Z_{>0}\times Q \mid x^2 - 2\Delta y^2 = 1\}\]
 of triples $(\Delta, x, y)$ of such coefficient-solution triples and the set
\[T = \{m \in Q \mid m(m+1) \in Q\}\subset Q\]
of elements $m$ that describe all $B$-smooth twins.
The correspondence is given in the following lemma, which is simply a reformulation of the proof of Theorem~\ref{thm:lehmer} as given in \cite{lehmer}.
\begin{lemma}\label{lem:twin_pell_corr}
Let the sets $P$ and $T$ be defined as above. The map 
$\delta: P \rightarrow T,\ (\Delta, x, y) \mapsto m = (x-1)/2$ 
is a bijection.
\end{lemma}
\begin{proof}
Let $(\Delta, x, y) \in P$. As noted in Remark~\ref{rem:xoddyeven}, $x$ must be odd and $y$ must be even. Hence both $m = (x-1)/2$ and $y/2$ are integers, as is $m+1 = (x+1)/2$. This means that $m\cdot(m+1) = (x^2-1)/4 = 2\Delta (y/2)^2 \in Q$ because $\Delta, y \in Q$, showing that $\delta(\Delta, x, y)$ lies in $T$.

Conversely, given $m\in T$ representing the twin $B$-smooth pair $(m, m+1)$, we have $m(m+1) \in Q$ and we can uniquely write 
$m(m+1) = 2q_1^{\alpha_1}\cdot q_2^{\alpha_2}\cdot \ldots \cdot q_t^{\alpha_t}$,
where $\alpha_i = \epsilon_i + 2\beta_i$ for $\epsilon_i\in \{0,1\}$ and integers $\alpha_i,\beta_i \ge 0$. Setting $x=2m+1$, $y = 2q_1^{\beta_1}\cdot q_2^{\beta_2}\cdot \ldots \cdot q_t^{\beta_t}$ and $\Delta = q_1^{\epsilon_1}\cdot q_2^{\epsilon_2}\cdot \ldots \cdot q_t^{\epsilon_t}$, we see that $x^2-1 = 4m(m+1) = 2\Delta y^2$. Therefore, given $m$, we obtain $x=2m+1$ and unique $\Delta$ and $y$ by the square-free and square parts of $4m(m+1)$ via $x^2 - 1 = 2\Delta y^2$, so we can define $\delta^{-1}$ by $\delta^{-1}(m) = (\Delta, x, y)\in P$. We have $\delta\circ\delta^{-1}(m) = m$ for all $m\in T$ and $\delta^{-1}\circ\delta(\Delta, x, y) = (\Delta, x, y)$ for all $(\Delta, x, y) \in P$, which implies the lemma.\qed
\end{proof}
Any solution $(x, y)$ to Equation~\eqref{eqn:twin_pell} is the $n$-th solution for some positive integer $n$ as defined in Theorem~\ref{thm:Pell_Solutions}. Lehmer~\cite[Thm.~1]{lehmer} showed that $n \leq M = \max\{3, (q_t+1)/2\}$. Therefore, the set $P$ is finite. 
An algorithm to compute all twin $B$-smooth pairs is thus to enumerate all coefficients $\Delta\in Q'$, compute the fundamental solution of the corresponding Pell equation $x^2 - 2\Delta y^2 = 1$ and then to compute all $n$-th solutions from that for $n\leq M$. To check whether any of those solutions yields a pair of twin $B$-smooth integers, it is sufficient to check whether the $y$-component of the solution is $B$-smooth. Clearly, this requires solving $2^t - 1$ Pell equations and quickly becomes infeasible for even moderate values of $B$.

\subsection{Twin smooth integers from $n$-th solutions for $n>1$}\label{sec:successive}
Next we consider for which $n$ the $n$-th solution to a Pell equation leads to a twin smooth pair via the above correspondence. We start by showing that if the $n$-th solution leads to a twin smooth pair then so does the fundamental solution. In other words, an $n$-th solution for $n>1$ can only correspond to a twin smooth pair if the corresponding fundamental solution corresponds to a twin smooth pair.
\begin{lemma}\label{lem:nth-to-1st-pair}
Let $\Delta\in Q'$, $(x_1, y_1)$ be the fundamental solution and $(x_n, y_n)$ the $n$-th solution to $x^2 - 2\Delta y^2 = 1$. Let $m_1 = (x_1-1)/2$ and $m_n = (x_n - 1)/2$. If $m_n \in T$, then so is $m_1$.
\end{lemma}
\begin{proof}
Since $m_n \in T$, we have $\delta^{-1}(m_n) = (\Delta, x_n, y_n) \in P$ and therefore, $y_n \in Q$. We have seen in \S\ref{sec:pell} that $y_1 \mid y_n$, so $y_1 \in Q$, which means that $(\Delta, x_1, y_1) \in P$. Then $ m_1 = \delta(\Delta, x_1, y_1) \in T$. \qed 
\end{proof}
This lemma shows that when solving Pell equations to find twin smooth pairs, an equation can be discarded if the fundamental solution does not yield a twin smooth pair and there is no need to check any of the successive solutions $(x_n, y_n)$, $n>1$ for that particular value of $\Delta$. From now on, let $m_n = (x_n-1)/2$ for $n\geq 1$, where $(x_n, y_n)$ is the $n$-th solution to the Pell equation~\eqref{eqn:twin_pell}.

Whenever the fundamental solution $(x_1, y_1)$ yields a twin smooth pair $(m_1, m_1+1)$, successive solutions might yield one as well. To decide whether for a fixed $(\Delta, x_1, y_1) \in P$ we also have $(\Delta, x_n, y_n)\in P$, it is sufficient to check whether $U_n \in Q$ because  then $y_n = U_ny_1 \in Q$. For that purpose, we can simply write $U_n=v_n(m_1)$ as a polynomial in $m_1$, where $v_n(x) = u_n(2x+1)$. Then $U_n = u_n(x_1) = u_n(2m_1+1) = v_n(m_1)$. These expressions for $n\leq 12$ are shown here:
\begin{align*}
v_1(m) & = 1,\\
v_2(m) & = 2(2m + 1),\\
v_3(m) & = (4m + 1)(4m + 3),\\
v_4(m) & = 4(2m + 1)(8m^2 + 8m + 1),\\
v_5(m) & = (16m^2 + 12m + 1)(16m^2 + 20m + 5),\\
v_6(m) & = 2(2m + 1)(4m + 1)(4m + 3)(16m^2 + 16m + 1),\\
v_7(m) & = (64m^3 + 80m^2 + 24m + 1)(64m^3 + 112m^2 + 56m + 7),\\
v_8(m) & = 8(2m + 1)(8m^2 + 8m + 1)(128m^4 + 256m^3 + 160m^2 + 32m + 1),\\
v_9(m) & = (4m + 1)(4m + 3)(64m^3 + 96m^2 + 36m + 1)(64m^3 + 96m^2 + 36m + 3),\\
v_{10}(m) & = 2(2m + 1)(16m^2 + 12m + 1)(16m^2 + 20m + 5)(256m^4 + 512m^3 + 304m^2 + 48m + 1),\\
v_{11}(m) & = (1024m^5 + 2304m^4 + 1792m^3 + 560m^2 + 60m + 1)\\
& \quad\cdot (1024m^5 + 2816m^4 + 2816m^3 + 1232m^2 + 220m + 11),\\
v_{12}(m) & = 4(2m + 1)(4m + 1)(4m + 3)(8m^2 + 8m + 1)(16m^2 + 16m + 1)\\
& \quad\cdot (256m^4 + 512m^3 + 320m^2 + 64m + 1).
\end{align*}
For the $n$-th solution to give a twin $B$-smooth pair, $U_n = v_n(m_1)$ must be smooth, where $(m_1, m_1+1)$ is the twin smooth pair given by the fundamental solution. This means that each of the irreducible terms in the polynomial representation of $v_n$ needs to evaluate to a $B$-smooth value at $m_1$. In the same way as above, we may write polynomials for the potential twin-smooth pair $(m_n, m_n+1)$ corresponding to the $n$-th solution, namely, we have that
\begin{equation}\label{eqn:mn}
m_n = (p_n(2m_1+1)-1)/2.
\end{equation}
In the proof of Lemma~\ref{lem:twin_pell_corr} we have seen that $m_1(m_1+1) = (x_1^2-1)/4 = 2\Delta (y_1/2)^2$ and $m_n(m_n+1) = (x_n^2-1)/2 = 2\Delta (y_n/2)^2$. Since $y_n = y_1U_n$, this means that 
\begin{equation}\label{eqn:mnmnp1}
    m_n (m_n+1) = m_1 (m_1+1) U_n^2 = m_1 (m_1+1) v_n(m_1)^2.
\end{equation}
The following lemma gives a more specific description about how $m_1$ and $m_1+1$ divide $m_n$ and $m_n+1$.
\begin{lemma}
Let $\Delta\in Q'$, $(x_1, y_1)$ be the fundamental solution and $(x_n, y_n)$ the $n$-th solution to $x^2 - 2\Delta y^2 = 1$. Let $m_1 = (x_1-1)/2$ and $m_n = (x_n - 1)/2$. Then $m_1 \mid m_n$. 
If $n$ is even, then $(m_1+1) \mid m_n$, if it is odd, then $(m_1+1) \mid (m_n+1)$. 
\end{lemma}
\begin{proof}
Note that $x_1^2-1 = 4m_1(m_1+1)$ and therefore, the explicit formula for $p_n$ in \S\ref{sec:pell} shows that $p_n(2m_1+1) \equiv 1 \mod (2m_1)$. It follows that $2m_1 \mid p_n(2m_1+1) - 1$ and thus $m_1 \mid (p_n(2m_1+1)-1)/2 = (x_n-1)/2 = m_n$.

We also see that $p_n(2m_1+1) \equiv (-1)^n \mod (2(m_1+1))$ such that $2(m_1+1) \mid p_n(2m_1+1)-1$ if $n$ is even and $2(m_1+1) \mid p_n(2m_1+1)+1$ if $n$ is odd.
\qed
\end{proof}

Below are the expressions for $(m_n, m_n+1)$ with $1\leq n \leq 12$.
\begin{align*}
(m_2, m_2+1) & = (4 m (m + 1),\ (2m + 1)^2),\\
(m_3, m_3+1) & = (m (4m + 3)^2,\ (m + 1) (4m + 1)^2),\\
(m_4, m_4+1) & = (16 m (m + 1) (2m + 1)^2,\ (8m^2 + 8m + 1)^2),\\
(m_5, m_5+1) & = (m (16m^2 + 20m + 5)^2,\ (m + 1) (16m^2 + 12m + 1)^2),\\
(m_6, m_6+1) & = (4 m (m + 1) (4m + 1)^2 (4m + 3)^2,\ (2m + 1)^2 (16m^2 + 16m + 1)^2),\\
(m_7, m_7+1) & = (m (64m^3 + 112m^2 + 56m + 7)^2,\ (m + 1) (64m^3 + 80m^2 + 24m + 1)^2),\\
(m_8, m_8+1) & = (64 m (m + 1) (2m + 1)^2 (8m^2 + 8m + 1)^2,\\
 & \qquad (128m^4 + 256m^3 + 160m^2 + 32m + 1)^2),\\
(m_9, m_9+1) & = (m (4m + 3)^2 (64m^3 + 96m^2 + 36m + 3)^2,\\
 & \qquad (m + 1) (4m + 1)^2 (64m^3 + 96m^2 + 36m + 1)^2),\\
(m_{10}, m_{10}+1) & = (4 m (m + 1) (16m^2 + 12m + 1)^2 (16m^2 + 20m + 5)^2,\\
 & \qquad (2m + 1)^2 (256m^4 + 512m^3 + 304m^2 + 48m + 1)^2),\\
(m_{11}, m_{11}+1) & = (m (1024m^5 + 2816m^4 + 2816m^3 + 1232m^2 + 220m + 11)^2,\\
 & \qquad (m + 1) (1024m^5 + 2304m^4 + 1792m^3 + 560m^2 + 60m + 1)^2),\\
(m_{12}, m_{12}+1) & = (16 m (m + 1) (2m + 1)^2 (4m + 1)^2 (4m + 3)^2 (16m^2 + 16m + 1)^2,\\
 & \qquad (8m^2 + 8m + 1)^2 (256m^4 + 512m^3 + 320m^2 + 64m + 1)^2).
\end{align*}

\begin{remark}
Lemma~\ref{lem:pnun}(b) shows that as soon as $n$ has an odd factor, the polynomial $p_n(x)$ used to represent $x_n = p_n(x_1)$ is reducible. This also means that $2m_n + 1 = p_n(2m_1+1)$ as a polynomial in $m_1$ is reducible and cannot assume a prime value for any $m_1$. Only $p_{2^k}(x)$ is irreducible for any $k\geq 1$ according to Lemma~\ref{lem:pnun}(e). Therefore, when we are interested in finding twin smooth pairs $(m_n, m_n+1)$, where $p = m_n + (m_n+1)$ is prime, only those solutions need to be checked where $n = 2^k$ is a power of $2$.
\end{remark}

Since pairs $(m_n, m_n+1)$ are polynomial expressions in the corresponding $m_1$, one can easily obtain bounds on the size of $m_n$ in terms of the size of $m_1$. For example, equations \eqref{eqn:mn} and \eqref{eqn:mnmnp1} allow to deduce a simple lower bound for $m_n$ in terms of the leading coefficient of $p_n$ or $v_n$. We obtain the following lemma.
\begin{lemma} \label{lemma:nth_solution_size}
Let $\Delta\in Q'$, $(x_1, y_1)$ be the fundamental solution and $(x_n, y_n)$ the $n$-th solution to $x^2 - 2\Delta y^2 = 1$, and let $m_1 = (x_1-1)/2$, $m_n = (x_n - 1)/2$, $U_n = y_n/y_1$ and $v_n(x) \in \Z[x]$ the above polynomial such that $U_n = v_n(m_1)$. Then
\[
m_1 v_n(m_1) \leq m_n < (m_1+1) v_n(m_1).
\]
It follows that $m_n > 4^{n-1}m_1^n$.
\end{lemma}
\begin{proof}
We use Equation~\eqref{eqn:mnmnp1} to see that $m_n^2 < m_n (m_n+1) = m_1 (m_1+1) v_n(m_1)^2 < (m_1+1)^2 v_n(m_1)^2$, which yields $m_n < (m_1+1) v_n(m_1)$. Similarly, $(m_n+1)^2 > m_n (m_n+1) = m_1 (m_1+1) v_n(m_1)^2 > m_1^2 v_n(m_1)^2$ and therefore $m_n \geq m_1 v_n(m_1)$. Because $v_n(x) = u_n(2x+1)$, $\deg(u_n) = n-1$ and the leading coefficient of $u_n$ is $2^{n-1}$ according to Lemma~\ref{lem:pnun}, the leading coefficient of $v_n$ is $(2^{n-1})^2 = 4^{n-1}$. The explicit formulas for $u_n$ in Section~\ref{sec:pell} show that all coefficients of $v_n(x)$ are positive and therefore, $v_n(m_1) > 4^{n-1}m_1^{n-1}$, which concludes the proof.\qed
\end{proof}

The last inequality in the lemma provides an upper bound on $m_1$ if we are given an upper bound on $m_n$. Assume, we are interested in twin smooth integers given by $m_n$ corresponding to an $n$-th solution to the Pell equation with $n>1$ with an upper bound $m_n < R$. Such pairs must then come from a fundamental solution pair given by $m_1$ such that $4^{n-1}m_1^n < R$, or, equivalently, $m_1 < \frac{1}{4}(4R)^{1/n}$.
For example, to find twin smooth integers where $m_n$ has at most $b$ bits, we have $R = 2^b$, i.e. $m_1 < R_b :=  2^{(b-2n+2)/n}$. Table~\ref{tbl:m1bounds} shows concrete values for the maximal possible bitsize of $m_1$ when $b=256$ and $n\in \{1,2,\dots,12\}$.

\begin{table}[]
    \centering
    \renewcommand{\tabcolsep}{0.2cm}
    \renewcommand{\arraystretch}{1.5}
    \begin{tabular}{c|cccccccccccc}
        \toprule
        $n$ &  $1$ & $2$ & $3$ & $4$ & $5$ & $6$ & $7$ & $8$ & $9$ & $10$ & $11$ & $12$\\
        $\lceil 258/n \rceil - 2$ & $256$ & $127$ & $84$ & $63$ & $50$ & $41$ & $35$ & $31$ & $27$ & $24$ & $22$ & $20$\\
        \bottomrule
    \end{tabular}
    \vspace{0.3cm}
    \caption{Maximal bitsizes $\lceil\log(R_{256})\rceil = \lceil 258/n \rceil - 2$ of $m_1$ such that the corresponding $m_n$ has at most $256$ bits for $n = 1, 2, \dots, 12$.}
    \label{tbl:m1bounds}
\end{table}

Further we have the following absolute upper bound on $n$ for which we can have twin smooth pairs corresponding to the $n$-th solution to a Pell equation (see \cite{lehmer}, Theorem 1).

\begin{theorem}
If $(m, m+1)$ is a pair of twin $B$-smooth integers corresponding to the $n$-th solution to a Pell equation, then $n \leq \frac{B+1}{2}$.
\end{theorem}

This upper bound is in many ways much weaker than the results of Lemma~\ref{lemma:nth_solution_size} because Lemma~\ref{lemma:nth_solution_size} shows that for large $n$, the $n$-th solutions to Pell equations under our bound all correspond to fundamental solutions that are exponentially smaller. Together, this suggests an algorithm for finding all smooth pairs in an interval corresponding to $n$-th solutions to Pell equations for large $n$. 
The algorithm works as shown in Algorithm~\ref{algorithm}.

\begin{algorithm}
\begin{algorithmic}[1]
\State Fix a lower bound $s$ on $n$.
\For {$n \in \{s, s+1, \dots, \frac{B+1}{2}$\}}
    \State Compute $T = \lceil M/n \rceil - 2$ as in Table~\ref{tbl:m1bounds}.
    \For {$w \leq T$}
        \If {$(w, w+1)$ is a twin $B$-smooth pair}
            \If {$v_n(w)$ is $B$-smooth}
                \State Record $(m_n, m_n+1)$ as a twin $B$-smooth pair.
            \EndIf
        \EndIf
    \EndFor
\EndFor
\end{algorithmic}
\caption{Algorithm to search for twin $B$-smooth pairs of size at most $M$ bits.}\label{algorithm}
\end{algorithm}

%\begin{enumerate}
%    \item Fix a lower bound $s$ on $n$.
%    \item Loop through all $n \in \mathbb{N}, s \leq n \leq \frac{B+1}{2}$. For each $n$:
%    \item Compute, as in the table above, $T = \lceil M/n \rceil - 2$ where $M$ is an upper bound on the smooth pairs we want to find.
%    \item Search through all integers $w \leq T$, and check if $(w, w+1)$ is a twin smooth pair. If so, check whether $v_n(w)$ is smooth, in which case $(m_n, m_n+1)$ is also a twin smooth pair. Then, record this pair.
%\end{enumerate}

The key benefit of this algorithm is that it finds smooth pairs of size roughly $2^M$ while only searching through integers up to $2^T$, an exponentially smaller range. In our running example of searching for pairs in the range $[2^{240}, 2^{256}]$, we ran the above algorithm with a lower bound of $s = 6$. Then the most costly iteration above with $n = 6$ only required checking approximately $2^{41}$ possibilities for $w$. Notice further that as $n$ gets larger, the possibilities for $w$ grow extremely restricted, and the algorithm terminates quickly. The overall runtime is dominated by the iteration with $n = s$.

For example, our run with $s = 6$ and a maximum of $M = 2^{275}$ found the following smooth pairs. A $16759$-smooth, $245$-bit pair $(m, m+1)$ with
\begin{align*}
m & = 44746808406030847930450201970587971020922341276429366152081686798603000000\\
  & = 2^6 \cdot 3^4 \cdot 5^6 \cdot 7^5 \cdot 11^2 \cdot 13^2 \cdot 41 \cdot 43^2 \cdot 53 \cdot 97 \cdot 241^2 \cdot 337 \cdot 509 \cdot 673 \cdot 4703^2\\ 
  & \quad \cdot 5981 \cdot 9413^2 \cdot 13669^2 \cdot 16759^2,\\
m+1 & = 31^2 \cdot 157^4 \cdot 181^2 \cdot 251^2 \cdot 349^2 \cdot 359^2 \cdot 457^2 \cdot 1427^2 \cdot 2617^2 \cdot 9109^2 \cdot 9649^2 \cdot 10253^2\\
\end{align*}
and a $24551$-smooth, $260$-bit pair with
{\small
\begin{align*}
m & = 1248045507865502270977250845951694434798578493856490782548653674169732908101560\\ 
 & = 2^3 \cdot 3^5 \cdot 5 \cdot 7^2 \cdot 17 \cdot 19 \cdot 29 \cdot 31 \cdot 43^2 \cdot 53^2 \cdot 149^2 \cdot 211^2 \cdot 227^2 \cdot 233 \cdot 827 \cdot 919^2 \cdot 2659 \cdot 4723^2\\
 & \quad \cdot 6907^2 \cdot 10831 \cdot 16691^2 \cdot 24551,\\
m+1 & = 11^2 \cdot 13^2 \cdot 23^2 \cdot 71^2 \cdot 107^2 \cdot 263^2 \cdot 587^2 \cdot 1021^2 \cdot 4297^2 \cdot 6491^2 \cdot 7309^2 \cdot 8089^2 \cdot 9049^2 \cdot 19009^2.
\end{align*}}
This second pair is larger than any previously known for its smoothness bound. 

\subsection{Lower order solutions}

The above algorithm finds all twin-smooth pairs corresponding to $n$-th solutions for $n$ large. In order to complete our approach, we now consider algorithms for finding pairs corresponding to Pell equation solutions with $n$ small. In this case, the bounds from Lemma~\ref{lemma:nth_solution_size} are too weak, and the algorithm from the previous section is intractable. Therefore, we search for such smooth pairs directly by solving Pell equations via continued fractions to obtain the fundamental solution, and then checking the successive solutions for small $n$ with our polynomials $v_n$.

The obvious difficulty with the approach is that there are exponentially many Pell equations which could have corresponding smooth pairs, and we cannot solve all of them. While we believe (see discussion) that this difficulty may be intrinsic, and that finding all pairs may not be tractable, we can still cleverly choose which Pell equations to solve in order to find some large twin-smooth pairs.

Recall that each of the $2^t - 1$ Pell equations $x^2 - Dy^2 = 1$ with potential corresponding smooth pairs has a unique coefficient $D = 2 \Delta$ where $\Delta$ is a smooth, square-free integer and $\Delta \neq 2$. Given any particular $\Delta$, solving the corresponding Pell equation can be done extremely quickly (see \ref{sec:cf}). Choosing which Pell equations to solve means choosing the values of $\Delta$. Generally, the difficulty is in finding $\Delta$s for which the size of the solution to the Pell equation is in the desired range. Fundamental solutions to Pell equations can be extremely large (exponential in $\Delta$), and such large solutions $(x,y)$ are very unlikely to have smooth $y$, and are thus not helpful to our search. It is therefore beneficial to choose coefficients $\Delta$ for which we expect the resulting pairs to generally be small.

\subsection{The smallest coefficient heuristic}\label{sec:smallCoeff}
A very simple heuristic for choosing $\Delta$ is to choose the smallest values possible. This is beneficial for a simple combinatorial reason. Note that if $(x,y)$ is a solution to the Pell equation
$$x^2 - 1= 2 \Delta y^2$$
then $y = \sqrt{\frac{x^2 - 1}{2 \Delta}}$ is an integer, and as we have seen in Section~\ref{sec:twins}, it must be smooth if $(x,y)$ corresponds to a smooth pair. Now, if $x$ lies in a certain range, such as $[2^{240}, 2^{256}]$ in our running example, then a small $\Delta$ means that there is a larger range for possible values of $y$. Thus we expect that having smaller $\Delta$ increases the probability that some $y$ solving the Pell equation will meet our criteria. Assuming that this is the only property of $\Delta$ affecting the chance of having a solution in the desired range, the probability of having a solution scales like $\frac{1}{\sqrt{\Delta}}$. Since $\Delta$ can range from $3$ to the square of the maximum bound on the size of smooth pairs we are interested in ($2^{512}$ in our example), the effect of this combinatorial heuristic is very strong.

This reasoning agrees with experiments. In our search for pairs in the range $[2^{240}, 2^{256}]$, we solved all Pell equations with $\Delta \leq 2^{44}$. This computation ran in parallel on a 64-core computer for 3 days and found a $19949$-smooth, $215$-bit pair $(m, m+1)$ corresponding to a fundamental solution with
\begin{align*}
m & = 51963397732665557125190357543988479960188331933248699616266017360\\
& = 2^4 \cdot 3^5 \cdot 5 \cdot 7 \cdot 11^2 \cdot 647 \cdot 911 \cdot 919^2 \cdot 1103^2 \cdot 2099^2 \cdot 2203^2 \cdot 2423^2 \cdot 5279^2 \cdot 8641^2 \cdot 19949,\\
m+1 & = 19^2 \cdot 23^2 \cdot 47^2 \cdot 277^2 \cdot 359^2 \cdot 419^2 \cdot 541^2 \cdot 887^2 \cdot 1993^2 \cdot 4549^2 \cdot 4813^2 \cdot 12721^2.
\end{align*}
While this pair is comparatively small to many others found, it is different in that it corresponds directly to a fundamental solution. Certain cryptographic applications require smooth pairs $(m, m + 1)$ for which the sum $m + (m + 1) = 2m + 1$ is prime. This property is not achievable by pairs corresponding to $n$-th solutions unless $n = 2^k$ for some integer $k$ by Lemma \ref{lem:pnun}.

Checking the $n$-th solutions of the pairs found using the previously discussed methods for $n \in \{1, 2, \dots, 12\}$ also found a pair
{\small
\begin{align*}
m & = 13751398930221343029252446890555634360426232035770955214339102967231667741706727080\\
 & =  2^3 \cdot 3^5 \cdot 5 \cdot 59 \cdot 103^2 \cdot 113^2 \cdot 1697 \cdot 2381 \cdot 2383^2 \cdot 2399^2 \cdot 3623^2 \cdot 4733 \cdot 9151^2 \cdot 10607^2\\
 & \quad \cdot 16267 \cdot 18059^2 \cdot 18289 \cdot 23603,\\
m+1 & = 11^4 \cdot 17^2 \cdot 53^2 \cdot 83^2 \cdot 109^2 \cdot 227^2 \cdot 263^2 \cdot 347^2 \cdot 373^2 \cdot 599^2 \cdot 2341^2 \cdot 7883^2 \cdot 10223^2 \cdot 10883^2 \cdot 12511^2,
\end{align*}}

which has $273$ bits and is $23603$-smooth. It corresponds to the sixth solution to a Pell equation. This pair, though larger than the target cryptographic size of 256 bits, is strong evidence for the efficiency of this algorithm. It is both larger and more smooth than the second pair presented in \ref{sec:successive}, and is far larger than any previously known for its smoothness bound. This pair is also a nice example of how the methods in \ref{sec:successive} can be combined well with other algorithms to yield $B$-smooth pairs larger than either algorithm could individually.

\subsection{The small primes heuristic}\label{sec:smallPrime}

While the above heuristic of choosing small values of $\Delta$ works well, a limitation of it is that there are relatively few small $\Delta$, and the probability of each Pell equation resulting in smooth pairs in the desired range decreases strongly as $\Delta$ increases. Furthermore, once we relax the restrictions on the maximum size of $\Delta$, exhausting all of its possible values becomes computationally infeasible. Another heuristic that narrows the search space is to choose values of $\Delta$ with smaller prime factors.

Among coefficients $\Delta$ of roughly the same size, those having more (and hence on average smaller) prime factors correspond to Pell equations that are more likely to result in smooth pairs in the desired range. We again give a combinatorial argument. If $p \mid \Delta$, then from the Pell equation, we have
$$x^2 - 1 \equiv 0 \pmod p \iff (x - 1) (x + 1) \equiv 0 \pmod p \iff x \equiv \pm 1 \pmod p$$
because the integers modulo $p$ form a field. Suppose that $\Delta$ is divisible by $n$ primes, say $\Delta = p_1 \cdot p_2 \cdot \ldots \cdot p_n$. Any integer $x$ that solves the Pell equation must satisfy the above condition, so we can associate to it a sequence $s \in \{-1, +1\}^n$, where the $i$th element of the sequence is $x \text{ mod } p_i$.

Now suppose we fix such a sequence $s$, and consider all $x$ in our desired range that are associated to $s$. By the Chinese remainder theorem, there is exactly one such $x$ in every interval of length $\Delta$. This gives further evidence for the above heuristic of choosing $\Delta$ small, because there are more $x$s in our desired range with the potential of solving the Pell equation.

Further, notice that $x$ in a solution to the Pell equation with coefficient $\Delta$ can correspond to any of the length $n$ sequences. Since there are $2^n$ such sequences and there are (up to $\pm 1$) the same number of $x$'s in the range corresponding to each, the implicit search space in which we look for solutions when solving Pell equations with coefficients $\Delta$ of a specified size is much larger when $\Delta$ has more prime factors.

Small searches with this method over coefficients larger than $2^{42}$ yielded a $201$-bit, $29881$-smooth pair 
\begin{align*}
m & = 2317395102010090979961970844394697249269453177012017537196644\\
 & = 2^2 \cdot 13 \cdot 53^2 \cdot 113 \cdot 139^2 \cdot 269^2 \cdot 347^2 \cdot 509^4 \cdot 569^2 \cdot 743^2 \cdot 12823^2 \cdot 14149 \cdot 29881,\\
m+1 & = 3^9 \cdot 5 \cdot 23^2 \cdot 43^2 \cdot 167^2 \cdot 179^2 \cdot 227 \cdot 349^2 \cdot 613^2 \cdot 661^2 \cdot 2927^2 \cdot 4099^2 \cdot 6421^2.
\end{align*}
It corresponds to $\Delta = 13 \cdot 113 \cdot 14149 \cdot 29881$.
This pair is smaller than those given previously, which is expected given the increase in the sizes of $\Delta$ searched over. However, it is still close to cryptographic size, indicating that more precise heuristics could extend our methods to search effectively over large values of $\Delta$.

\section{Discussion}

We have seen how solutions to Pell equations give both a theoretical and practical approach to finding twin smooth integers, including several new algorithms that have produced new pairs. We conclude by discussing two observations for exploration in further research.

Based on experiments, our algorithms are more likely to find smooth pairs $(m, m+1)$ with $2m + 1$ composite, which is undesirable for some applications. Based on Lemma \ref{lem:pnun}, the only pairs that can sum to a prime are those corresponding to $n$-th solutions with $ n = 2 ^ k $ for integer $k$. Generally such pairs are easy to exhaust with the algorithms given in \ref{sec:successive} for $k$ even moderately large; for our example we found that no pairs in $[2^{240}, 2^{256}]$ exist with $k \geq 3$. Therefore, it is interesting to consider whether we can improve our search through fundamental solutions to directly find large pairs more often.

One approach is to consider the structure of the continued fractions expansions. Because it was not useful to our algorithms, we have omitted discussion of many remarkable structures in the continued fractions of square roots, such as the periodicity of the quotients. Because the size of fundamental solutions is determined by the length of the period and size of the quotients of $\sqrt{\Delta}$, one idea is to choose coefficients $\Delta$ by fixing the length of quotients in the corresponding continued fraction expansion.

Unfortunately, while this approach guarantees that we have pairs $(x,y)$ satisfying the Pell equation $x^2 - 1 = 2\Delta y^2$ with $x$ in a given range, however, it does not guarantee smoothness of $y$ or $D$. It is interesting to consider whether a clever way of choosing the period lengths and continued fractions quotients could result in a high probability of $y$ and $D$ being smooth, thus resulting in an efficient algorithm for finding twin smooth pairs. This is a promising approach because of the strong and well-studied structure of continued fractions. An overview of the theory of continued fractions can be found in \cite{jacobson-williams}, chapters 3 and 6.

The tension described above between searching for solutions to Pell equations (or equivalently smooth pairs) with a fixed smoothness bound and searching for solutions (or smooth pairs) in a given size range is common across many algorithms for finding twin smooth integers. This difficulty seems intrinsic, as the problem of predicting the size of the solution to a Pell equation is related to questions about computing class numbers of quadratic fields, which remain highly mysterious.

Solving the Pell equation is equivalent to finding the fundamental unit in the associated real quadratic field $\mathbb{Q}(\sqrt{D})$. Because of the structure of units in its ring of integers, given the fundamental unit $\epsilon_D = t + u\sqrt{D}$, then either $x = t, y = u$ or $x = t^2 + u^2, y = 2tu$ solve the Pell equation $x^2 - D y^2 = 1$. The two cases arise from the fact that the fundamental unit can have norm $\pm 1$. For simplicity we focus on the first case; the analysis is analogous for the second. Since $x \approx \sqrt{D}y$, $\epsilon_D \approx 2x$.

Dirichlet's class number formula for quadratic fields relates the size of the fundamental unit to that of the fields' class number $h(D)$,

$$h(D) = \frac{\mathcal{L}(\chi, 1) \sqrt{D}}{2 \text{log}(\epsilon_D)}.$$

Based on both experimental evidence and theoretical bounds, $\mathcal{L}(\chi, 1)$ grows extremely slowly with $D$, see for example the following plot.

\includegraphics[scale = 0.8]{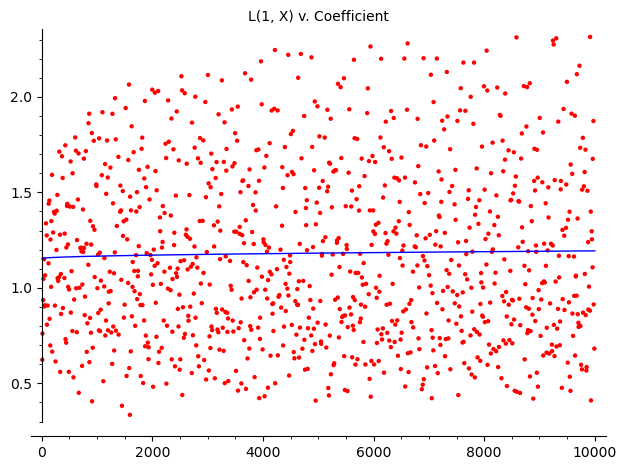}

Therefore if we want to find coefficients $D$ with corresponding $x \approx M$, we want

$$h(D) \approx \frac{\sqrt{D}}{\text{log}(M)}.$$

Since $M$ is fixed, choosing smooth, squarefree coefficients $D$ with Pell equation solutions in the right range is a problem equivalent to choosing $D$ with class number roughly $\sqrt{D}$. This problem is likely very difficult, since little is know about the sizes of class numbers of real quadratic fields. For example, Gauss conjectured in 1801 that there are infinitely many real quadratic fields with class number 1; this conjecture remains open. In the absence of new breakthroughs in predicting the sizes of class numbers, heuristics such as those presented in \ref{sec:smallCoeff} and \ref{sec:smallPrime} appear to be the best way to search for large twin smooth integers corresponding to fundamental solutions.

\bibliography{pell-twins}

\begin{thebibliography}{10}

\bibitem{CHM4SQISign}
Giacomo Bruno, Maria Corte-Real Santos, Craig Costello, Jonathan~Komada
  Eriksen, Michael Naehrig, Michael Meyer, and Bruno Sterner.
\newblock Cryptographic smooth neighbors.
\newblock Cryptology ePrint Archive, Paper 2022/1439, 2022.
\newblock \url{https://eprint.iacr.org/2022/1439}.

\bibitem{pte}
M.~Naehrig C.~Costello, M.~Meyer.
\newblock Sieving for twin smooth integers with solutions to the
  prouhet-tarry-escott problem.

\bibitem{Conrey}
J.~B. Conrey, M.~A. Holmstrom, and T.~L. McLaughlin.
\newblock Smooth neighbors.
\newblock {\em Experimental Mathematics}, 22(2):195--202, 2013.

\bibitem{b-sidh}
C.~Costello.
\newblock {B-SIDH:} supersingular isogeny {Diffie-Hellman} using twisted
  torsion.
\newblock In S.~Moriai and H.~Wang, editors, {\em {ASIACRYPT} 2020}, volume
  12492 of {\em Lecture Notes in Computer Science}, pages 440--463. Springer,
  2020.

\bibitem{sqisign}
L.~De Feo, D.~Kohel, A.~Leroux, C.~Petit, and B.~Wesolowski.
\newblock {SQISign}: Compact post-quantum signatures from quaternions and
  isogenies.
\newblock In Shiho Moriai and Huaxiong Wang, editors, {\em {ASIACRYPT} 2020},
  volume 12491 of {\em Lecture Notes in Computer Science}, pages 64--93.
  Springer, 2020.

\bibitem{jacobson-williams}
M.~J. Jacobson and H.~C. Williams.
\newblock {\em Solving the {Pell} Equation}.
\newblock Springer, 2009.

\bibitem{lehmer}
D.~H. Lehmer.
\newblock On a problem of {S}t{\"o}rmer.
\newblock {\em Illinois Journal of Mathematics}, 8(1):57--79, 1964.

\bibitem{lenstra02}
H.~W. Lenstra.
\newblock Solving the {Pell} equation.
\newblock {\em Notices Amer. Math. Soc.}, 49(2), 2002.

\bibitem{luca-najman-err}
F.~Luca and F.~Najman.
\newblock On the largest prime factor of $x^2-1$.
\newblock {\em Math. Comp.}, 80(273):429--435, 2011.

\bibitem{luca-najman}
F.~Luca and F.~Najman.
\newblock Errata to {"On the largest prime factor of $x^2-1$"}.
\newblock {\em Math. Comp.}, 83(285):337, 2014.

\bibitem{Rosen}
K.~H. Rosen.
\newblock {\em Elementary Number Theory and Its Applications}.
\newblock Addison Wesley Publishing Company, 1984.

\bibitem{stormer}
C.~St{\o}rmer.
\newblock Quelques th{\'e}or{\`e}mes sur l'{\'e}quation de {Pell}
  $x^2-dy^2=\pm1$ et leurs applications.
\newblock {\em Christiania Videnskabens Selskabs Skrifter, Math. Nat. Kl},
  (2):48, 1897.

\bibitem{unger}
J.~Unger.
\newblock Solving {Pell's} equation with continued fractions.
\newblock 2009.

\end{thebibliography}
\bibliographystyle{plain}

\section{Code}

The code can be found at \url{https://github.com/lightningwolf111/twin_smooth_integers}. Large searches were done using the parallelized C implementations and use the \href{https://gmplib.org/}{GMP} library. Experiments with many other algorithms can be found in the repository as well, and were implemented in python using the \href{www.sagemath.org}{SAGE} library.

\end{document}